%% file: art.tex
\documentclass{article}
\PassOptionsToPackage{hyphens}{url}\usepackage{hyperref}
\usepackage{authblk}
\pdfoutput=1
\usepackage{algorithmic}
\usepackage[english]{babel}
\usepackage{epsfig}
\usepackage{amsfonts,amssymb,amsmath,amsthm}
\usepackage{amssymb}
\usepackage{latexsym}
\usepackage{url}
\usepackage{amsmath}
\usepackage{thmtools}
\usepackage[numbers]{natbib}
\usepackage{comment}
\usepackage{tikz}
\usetikzlibrary{arrows, positioning,chains,fit,shapes,calc}
\usetikzlibrary{decorations.pathmorphing}
\usetikzlibrary{decorations.markings}
\usetikzlibrary{decorations.pathreplacing}
\usepackage{float}
\usepackage{graphicx}
\usepackage{mathtools}

\newcommand\red[1]{{\color{red} #1 }}
\newcommand\ar[1]{\vec{#1}}

\declaretheorem{assumption}
\declaretheorem{algorithm}
\declaretheorem{proposition}
\declaretheorem{theorem}

\newcommand{\CC}{\ensuremath{\mathbb{C}}}

\newcommand{\defn}[1]{\emph{#1}}
\DeclareMathOperator{\EN}{\mathbf{E}}

\DeclareMathOperator{\Bin}{Bin}
\DeclareMathOperator{\sol}{sol}
\DeclareMathOperator{\supp}{supp}
\DeclareMathOperator{\dEN}{\Delta}
\DeclareMathOperator{\EE}{\mathbb{E}}
\DeclareMathOperator{\tail}{src}
\DeclareMathOperator{\head}{dest}
\DeclareMathOperator{\edge}{edge}

\def\longleftrightarrowfill@{\arrowfill@\leftarrow\relbar\rightarrow}
\makeatother
\tikzset{
  text style/.style={
    sloped, 
    text=black
  }
}

\newcommand*\samethanks[1][\value{footnote}]{\footnotemark[#1]}

\begin{document}
\title{Monodromy Solver: Sequential and Parallel\thanks{Research of TD and AL is supported in part by NSF grant DMS-1719968.}}
\author{Nathan Bliss\thanks{University of Illinois at Chicago}   \hspace{7mm}  Timothy Duff\thanks{Georgia Tech}\\
Anton Leykin\samethanks \hspace{7mm} Jeff Sommars\samethanks[2]}

\maketitle

\begin{abstract} 
We describe, study, and experiment with an algorithm for finding all solutions of systems of polynomial equations using homotopy continuation and monodromy. 
This algorithm follows the framework developed in~\cite{firstmonodromypaper}
and can operate 
in the presence of a large number of failures of
the homotopy continuation subroutine.

We give special attention to parallelization and probabilistic analysis of a model adapted to parallelization and failures. Apart from theoretical results, we developed a simulator that allows us to run a large number of experiments without recomputing the outcomes of the continuation subroutine.
\end{abstract}


\maketitle

\section{Introduction}
\input{introduction.tex}

\section{Monodromy Solver framework}\label{basicMonodromy}
\input{basicMonodromy.tex}

\section{Definitions}\label{definitions}
\input{definitions.tex}

\section{Task selection via potential}\label{sec:task-selection}
\input{results.tex}

\section{Algorithm}\label{alg}
\input{algorithms.tex}

\section{Failure rate and threshold}\label{sec:threshold}
\input{thresh.tex}

\section{Experimental Results}\label{experiments}
\input{experiments.tex}

\section{Conclusion}\label{sec:conclusion}
\input{conclusion.tex}

\bibliography{art.bbl}

\end{document}

%% file: introduction.tex
Monodromy Solver (MS) is an algorithmic framework for solving parametric families of polynomial systems.  MS relies on numerical homotopy continuation methods \cite{Morgan87}, which are applicable in a very general setting, and monodromy (Galois group) action, which is specific to polynomial systems. The monodromy technique has been used successfully in \emph{numerical algebraic geometry} (for a good overview, see \cite{Sommese-Wampler-book-05}) mostly for high level tasks: for instance, numerical irreducible decomposition~\cite{SVW2001:monodromy} or Galois group computation~\cite{Leykin-Sottile:HoG,hauenstein2016numerical}.

\bigskip

The MS framework addresses the following basic problem:
\begin{quote}
Given a family of polynomial maps $F_p$, find \emph{all} solutions to $F_p=0$ for a generic value of $p$.
\end{quote}

\noindent Note that, given an ability to construct a complete solution set for a generic
value of the parameter, one can find all isolated solutions $F_p=0$ for an \emph{arbitrary} value of the parameter $p$ by using a coefficient-parameter homotopy~\cite[\S
7]{Sommese-Wampler-book-05}.

Apart from MS,  current methods of polynomial system solving via homotopy continuation include polyhedral approaches
\cite{HuberSturmfels:PolyhedralHomotopies, Verschelde-Verlinden-Cools},
total-degree and multihomo-geneous-degree homotopies \cite{WAMPLER1992143},
regeneration \cite{hauenstein2011regeneration}, and various other more specialized methods.

Most methods of homotopy continuation are embarrassingly parallel, in that homotopy paths
can be tracked completely independently. Literature on parallelism in relation to homotopy continuation includes \cite{Gunji2006,
Harimoto:1987:GHA:645818.669226, 1488691, Leykin:2007:CMV:1278177.1278195,
Ley-Ver-new-monodromy-05, Leykin2006, MORGAN19891339, 1690736}.

While an \defn{atomic task} of MS (one homotopy path track) is independent of another such task, this is only true for the tasks that are \emph{already} scheduled and being processed.

The scheduling algorithm, however, follows a probabilistic framework and (at every state when resources free up) attempts to find a task that maximizes the number of solutions known once this task and the current (in progress) tasks are complete. This has to be done using only partial knowledge of the outcome of the current tasks and, hence, implies a \emph{dependence} of the choice of a new candidate task on the current state of the algorithm.

In the context of the framework that allows multiple threads to carry out atomic tasks, we analyze the probabilistic model that results from the assumption of uniform randomness of correspondences induced by edges in an underlying graph (see \S\ref{basicMonodromy}). This is followed by analysis of a model that accounts for failures in the homotopy continuation subroutine. 

Last, but not least, we have implemented a simulator for the new algorithm that makes it possible to run fast experiments without rerunning the actual continuation subroutine over and over again. Using this simulator we conduct several computational experiments on both fabricated data using the probability distribution in our model and the data coming from the execution of homotopy continuation algorithms for a family of polynomial systems. This contribution is important for the further development of the MS framework, since our probabilistic assumptions are too simple to completely describe the random behavior in actual computations.

In \S\ref{basicMonodromy} we give a primer on MS using an example, and then define necessary terminology in \S\ref{definitions}. The pseudocode for the main algorithm appears in \S\ref{alg}.  A probabilistic model is analyzed in \S\ref{sec:task-selection} with a view towards designing a task selection strategy for our algorithm. The study of the threshold for completion depending on the rate of failures is in \S\ref{sec:threshold}. 
Finally, in~\S\ref{experiments}, we describe the implementation of the simulator, and use it to showcase the benefits of the new approach via several experiments.
A brief conclusion is in \S\ref{sec:conclusion}.


%% file: basicMonodromy.tex
For a family $F_p$, the MS approach treats different parameter choices $p_i$ as nodes in a graph, and by tracking along ``edges'' (i.e.\ homotopies) between them, seeks to populate the solution set for at least one node. Each of these homotopies is a \emph{coefficient-parameter homotopy},
\begin{equation}
H(t)=F_{(1+t)p_1+tp_2}, ~~~t\in[0,1],
\end{equation}
which tracks between the parameter choices $p_1$ and $p_2$. For generic $p$, the
number of roots of the system is constant, and following loops in the
graph permutes the roots. In the case when $F_p$ is linear in $p$, we may use a \emph{segment homotopy},
\begin{equation}\label{eq:homotopy}
H(t)=(1-t)\gamma_1F_{p_1} + t\gamma_2F_{p_2}, ~~~t\in[0,1],
\end{equation}
defined for generic $\gamma_1,\gamma_2\in\CC$. This gives us the ability to introduce multiple edges between two nodes, in hope that they would induce distinct maps on the solution sets.

As an example, suppose we want to know the roots of a generic univariate cubic
polynomial. Writing it as 
\begin{equation}
x^3 + ax^2 + bx + c,
\end{equation}
we set up a graph for three values $p_1,p_2,p_3\in\mathbb{C}^3$ of 
the coefficients.  It may help to visualize the family with one parameter: set $a=b=0$. Then we may just imagine a triangle embedded in the complex plane, i.e., the parameter space of $c$. This triangle is the \emph{homotopy graph} of Figure~\ref{fig:basicMonodromy} that lifts to the \emph{solution graph} above it.

\begin{figure}[H]
\centering
\begin{tabular}{lccccr}
\begin{minipage}{0.35\textwidth} 
  \begin{tikzpicture}[auto, every node/.style={circle,draw},
    every fit/.style={ellipse,draw,inner sep=-2pt,text width=1.75cm}]
    \node (a0) [fill=gray]  {};
    \node (a1) [below of=a0,fill=white] {};
    \node (a2) [below of=a1,fill=white] {};

    \node (b0) [xshift=0.45\textwidth,fill=white]  {};
    \node (b1) [below of=b0,fill=white] {};
    \node (b2) [below of=b1,fill=white] {};

    \node (c0) [xshift=0.2\textwidth,yshift=-0.25\textwidth,fill=white]    {};
    \node (c1) [below of=c0,fill=white] {};
    \node (c2) [below of=c1,fill=white] {};

    \draw[dashed] (a0) -- (b1); \draw[dashed] (a1) -- (b0); \draw[dashed] (a2) -- (b2);
    \draw[dashed] (c0) -- (b0); \draw[dashed] (c1) -- (b2); \draw[dashed] (c2) -- (b1);
    \draw[dashed] (a0) -- (c0); \draw[dashed] (a1) -- (c1); \draw[dashed] (a2) -- (c2);

    \newlength{\downshift}
    \setlength{\downshift}{-0.7\textwidth}
    \node (a) [yshift=\downshift,fill=white,label=below:{$p_1$}]  {};
    \node (b) [xshift=0.45\textwidth,yshift=\downshift,fill=white,label=below:{$p_2$}]  {};
    \node (c) [xshift=0.2\textwidth,yshift=\downshift-0.2\textwidth,fill=white,label=below:{$p_3$}]    {};
    \draw (a) -- (b);
    \draw (b) -- (c);
    \draw (a) -- (c);
  \end{tikzpicture}
\end{minipage}

&
&
&
&
&

\begin{minipage}{0.2\textwidth}
\hspace{-2.8cm}\includegraphics[width=5cm]{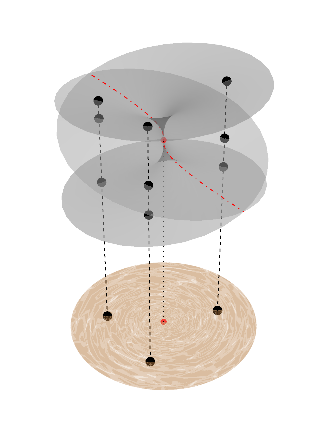}
\end{minipage}
\end{tabular}
\caption{The homotopy graph (left) and the solution graph (right)
viewed as a restriction of the 3-to-1 covering for $x^3 + c=0$.}
\label{fig:basicMonodromy}

\end{figure}
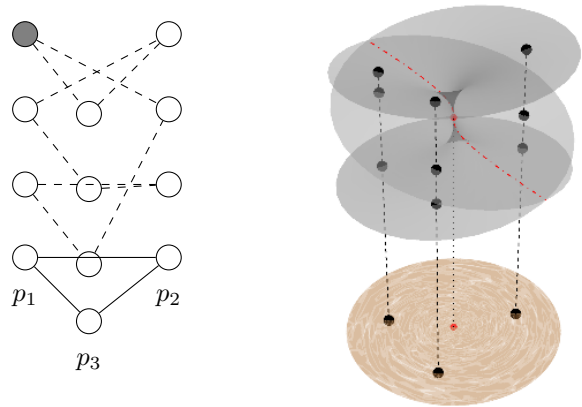

Assume that we know one solution (shaded in Figure~\ref{fig:basicMonodromy}) of $F_{p_1}$. Now continue this solution along edges of the solution graph. By doing so, we recover \emph{all three} solutions at one of the nodes. 
As long as the action of the \defn{monodromy group} (see \cite{firstmonodromypaper} for definition and discussion) is transitive, it is always possible to recover all solutions from one by following along the edges of a graph that is sufficiently large and sufficiently general.

In our very simple example there is always a unique choice for the next edge to track along; in general, this will not be so. The fact that only a subgraph of the solution graph is known at any point of the algorithm complicates the selection of the next (homotopy continuation) atomic task further. Hence, two parts of the MS approach are probabilistic: first, the homotopy graph is created at random; and second, the task selection procedure may either be random or designed to maximize the expectation of some \emph{potential} function (see \S\ref{sec:task-selection}) under a fixed probabilistic model. 

It has been shown in \cite{firstmonodromypaper} that given a simple probabilistic model, the expected number of edges in the solution graph for MS to succeed is \emph{linear} in the number of solutions. This bounds the number of continuation tasks to be carried out and---what seems to be the main reason for the practical success of MS---ties the overall complexity of the approach to the \emph{actual} number of the solutions, and not to some \emph{bound} that may be available a priori for a larger family of systems (e.g., bound  of B\'ezout or Bernstein-Khovanskii-Kouchnirenko).

Note that the Monodromy Solver framework does not specify a stopping criterion. For the discussion of possible stopping criteria see \S 3.2.2 and \S 3.2.3 of~\cite{firstmonodromypaper}.


%% file: definitions.tex
Let $G$ be a loopless multigraph with vertices $V=V(G)$ and edges $E = E(G)$.
Each vertex corresponds to a system $F_p$ specialized at parameters $p$ and is
associated with $d$ solutions---we refer to the vertex together with this
satellite data as a \emph{node.}  Each edge $e \in E$ connecting nodes $v_1$ and $v_2$ induces
a homotopy that establishes a bijective correspondence between the solutions of
the polynomial systems $F_{p_1}$ and $F_{p_2}$. We assume the following:
\begin{assumption}\label{randomAssumption}
Edges induce uniformly random correspondences.
\end{assumption}
\noindent In other words, we assume that all bijective correspondences that
could be induced by an edge connecting the solutions of $F_{p_1}$ and $F_{p_2}$
are equally likely.
This assumption allows us to simplify the probability calculations and postulate an effective task selection strategy described in \S\ref{sec:task-selection} especially when tracking multiple paths in parallel. See discussion of randomization in \S 5.1 of~\cite{firstmonodromypaper}.

In general, $e$ will refer to an edge and $\ar{e}$ will refer to a directed
edge (a pair of $e$ and a specified direction). A pair $t=(s,\ar e)$,
where $s\in S(v)$ belongs to the solution set $S(v)$ of the polynomial system
corresponding to $v = \tail{(\ar e)}$, represents a candidate for (one)
\emph{homotopy path track},  an atomic task that shall be performed by one
thread in a parallel algorithm.  

We fix the graph $G=(V,E)$ at the initialization stage. At a given \defn{state}
$x=(Q,C,A)$ of the algorithm we have the following.

\begin{itemize} 

\item A collection $Q$ of sets indexed by $v\in V$, where each $Q_v$ is
the subset of solutions at $v$ \defn{known} at this state.

\item A collection $C$ of sets indexed by $e \in E$. Each $C_e
\subset Q_v\times Q_w$ --- where $v$ and $w$ are the nodes $e$ connects --- 
is a partial one-to-one correspondence between subsets of
$Q_{v}$ and $Q_{w}$. We denote by $\pi_v$ and $\pi_w$ the projections from
$Q_v\times Q_w$ to $Q_{v}$ and $Q_{w}$, respectively.

\item A set $A = \{t_1,\dots,t_k\}$ of atomic (homotopy path tracking) tasks
currently being processed (using $k$ independent threads).

\end{itemize}
Given a state $x = (Q,C,A)$ we denote $Q(x):=Q$, $C(x):=C$, and $A(x):=A$. Note that in most states (in our basic framework, in all states but the initial state) one can determine $Q$ from $C$. We shall call a state $x$ \defn{idle} if $A(x) = \emptyset$. 
 
For an atomic task $t=(s,\ar e)\in A$, 
$\tail(t)$ and $\head(t)$ will refer to the source and destination vertices of 
$\edge(t) := \ar e$.

Prior to running $t,$ it is unknown which solution at $\head (t)$ will be found. We use the random variable $\sol_t$ to denote the outcome of running this task, conditioned on the current state. Likewise, $\sol_A$ will denote the random set of solutions known after running the tasks in $A$.
Suppose we know (or can estimate) the solution count for a generic system; refer to this (integer) count as $d$, the \defn{degree} of the problem.   Assumption~\ref{randomAssumption} implies that

\begin{equation}\label{singleThreadEV}
\Pr (\sol_t \notin Q_{\head (\ar e)})  = \frac{d-|Q_{\head(\ar e)}|}{d-|C_{e}|}
\end{equation}

Define $\EN(x)$ to be the expected  total number
of known solutions at all vertices after running \emph{all} tasks $t\in A$ to
completion. That is, if $y$ is the state after the completion, i.e., $A(y) = \emptyset$, 
then 
$$ \EN(x) = \sum_{v\in V(G)} \EN_v(x),\text{ with } \EN_v(x) :=
\EE(|Q_v(y)|), $$ where the (new) number of known points $|Q_v(y)|$ is perceived as
a random variable with expected value $\EE(|Q_v(y)|)$; state transition
probabilities are induced by Assumption~\ref{randomAssumption}.


%% file: results.tex
We intend to use either $\EN(Q,C,A \cup\{t\})$ to define a \emph{potential}
function driving our choice of the next task $t$ to append to $A$ once a
thread becomes available. The basic update rule is given below:
\begin{equation}
  \label{eq:increment}
  \EN_v (Q,C,A\cup \{t\}) = \EN_v (Q,C,A) + \Pr (\sol_t \notin \sol_A).
\end{equation}
This follows by a simple conditioning argument:
\begin{align*}
  \EN_v (Q,C,A\cup \{t\}) &=\EN_v (Q,C,A\cup \{t\} \mid \sol_t \in \sol_A) \Pr (\sol_t \in \sol_A)\\
&+ \EN_v (Q,C,A\cup \{t\} \mid \sol_t \notin \sol_A) \Pr (\sol_t \notin \sol_A) \\[0.3em]
  &=  \EN_v (Q,C,A) \Pr (\sol_t \in \sol_A) \\
  &+ \Big(1+\EN_v (Q,C,A)\Big) \Pr (\sol_t \notin \sol_A)\\[0.3em]
  &=  \EN_v (Q,C,A) + \Pr (\sol_t \notin \sol_A)\\
\end{align*}

\subsection{Potential given no path failures}
Since random homotopy paths stay away from the discriminant locus with probability $1$, it is natural to seek a ``smart'' task-selection strategy in the idealized setting when \emph{no failures} in homotopy tracking occur. 
The following proposition shows that $\EN_v(Q,C,A)$ can be computed recursively. 
\begin{proposition}\label{prop2}
  Let $v$ be a
vertex and $e$ an edge incident to $v$.  If $t$ is a candidate path track with $\head(t)=v$ and
$\edge(t)=\ar e$, then
\begin{eqnarray}\label{eq:potential}
  \begin{aligned}
    \EN_v(Q,C,A\cup \{t\}) & =  \EN_v(Q,C,A)\\
    & + \frac{d-\EN_v(Q,C,A)}{d-|C_e|-\#\{t' \in A : \edge(t')=\ar e\}}
  \end{aligned}
\end{eqnarray}
\end{proposition}
Thus, if we keep track of these expectations as we go, we may determine the potential of tracking a new thread without recomputing anything else.
\begin{proof}
Let $X$ denote the random variable that, conditioned on the idle state $(Q,C,\emptyset),$ counts the total number of solutions at $v$ after completing all tasks in $A.$ Noting equation~\eqref{eq:increment}, we have
\begin{eqnarray*}
  &&\Pr (\sol_t \notin \sol_A) =  \\
  &&=   \sum_{k\in \supp X} \displaystyle\frac{(d-k)}{d-|C_e|- \#\{t' \in A : \edge(t')=\ar e\}} \Pr (X =k) \\
  &&= \frac{d-\EN_v(Q,C,A)}{d-|C_e|-\#\{t' \in A : \edge(t')=\ar e\}}.
\end{eqnarray*}
\end{proof}

\subsection{Potential in the presence of failures}\label{sec:failures}
The failure of certain atomic tasks is an inevitable feature of any MS implementation: such failures may occur when paths verge too close to the locus of singular systems, and may be influenced by the aggressiveness of threshold settings in the underlying numerical software as well as various others factors. In anticipation of such failures, we consider the effects of failures in a simple probabilistic model generalizing the results of the previous section.

\begin{assumption}\label{failureAssumption}
We now assume that the \defn{probability of success} for every atomic task equals a global fixed constant $\alpha\in[0,1]$ and that formation of edge correspondences and all task failures are mutually independent events.
\end{assumption}

Let us emphasize a technical feature of this assumption---if we have $\edge (t) = \edge (t')$ and $\head (t) = \tail (t'),$ the tasks $t$ and $t'$ still fail independently. This lack of symmetry should be accounted for in any given state of the algorithm. Thus, we extend our definition of a \emph{state} $x=(Q,C,A,F)$ as follows:

\begin{itemize}

\item As before, $Q_v$ denotes the set of solutions known at $v\in V,$ each $C_{vw} \subset Q_v \times Q_w$ is a set of \emph{known, successful correspondences,} and $A$ is the set of current tasks.

\item Failures are indexed by \emph{directed edges.} For each $\ar e$, the set $F_{\ar e}$ consists of known solutions $s \in Q_{\tail(e)}$ such that the task $(s,\ar e)$ has completed with a failure. For $\alpha =1,$ we have $F_{\ar e}$ empty for all $\ar e$ and hence abbreviate $x=(Q,C,A).$
\end{itemize}
\begin{proposition}\label{prop_fail_pot}
With notation as in Proposition~\ref{prop2},
  \begin{eqnarray*}\label{eq:potentialfail}
    &&\EN_v(Q,C,A\cup \{t\},F)  =  \EN_v(Q,C,A,F) \, + \alpha \times  \\
    && \frac{\Big( d-\EN_v(Q,C,A,F)\Big) \Big(  1 - \EE \, \frac{\# F_{\edge (t)} + B}{d- \# C_e - \cdot \# \{ t' \in A \mid \edge (t') = \ar e \} + B} \Big)}{d - \# C_e - \# F_{\edge (t)} -  \# \{ t' \in A \mid \edge (t') = \ar e \}  }.
  \end{eqnarray*}
where $B$ a random variable with a binomial distribution: $B\sim \Bin (\#\{t' \in A : \edge(t')=\ar e\}, 1-\alpha ).$
\end{proposition}
\begin{proof}
Let $u = \tail(\edge(t)).$ We consider the following set-valued random variables whose state spaces are conditional on the idle state $(Q,C,\emptyset , F):$
\begin{itemize}
\item $X$ is the set of all solutions at $v$ which are known after completing all tasks in $A$--hence $\EE [X] = \EN_v (Q,C,A).$
\item $Y$ consists of all solutions at $v$ whose correspondences along $\ar{uv}$ have failed after completing all tasks in $A.$ Thus, the random variable $B:=(\# Y- \# F_{\edge (t)})$ has the desired binomial distribution.
\end{itemize}
Recalling~\eqref{eq:increment}, note that task $t$ yields a solution undiscovered by $A$ with probability $\alpha \cdot \Pr (\sol_t \notin \sol_A ).$  Moreover, we have
\[
\Pr (\sol_t \notin \sol_A ) = \displaystyle\frac{d-\EN_v (Q,C,A,F) - \EE [ \# (Y \cap X^c) ]}{d-\# C_e - \# F_{\edge (t)} - \# \{ t' \in A \mid \edge (t') = \ar e \} }
\]
Conditional on the event $(X=k,B=j),$ we have that
\[
\#Y = \# F_{\edge (t)} + j,
\]
but the intersection $Y \cap X^c$ still depends on
\[
d - (\# C_e +  \# \{ t' \in A \mid \edge (t') = \ar e \} -j)
\]
unknown correspondences. Assumption~\ref{randomAssumption} implies that the conditional distribution may be generated as follows:
\begin{itemize}
\item[1)] For each solution at $u$ which is known to fail along $\ar{uv}$ after completing all tasks in $A,$ the corresponding solutions in $Y$ are drawn uniformly without replacement from the $(d - \# C_e -  \# \{ t' \in A \mid \edge (t') = \ar e \} +j)$ solutions at $v$ without correspondences.
\item[2)] Declare each solution in $Y \cap X^c$ to be a ``success.''
\end{itemize}
Hence the conditional expectation of the number of ``successes'' is given by the mean of a hypergeometric distribution:
\begin{eqnarray}
\EE \left( \# (Y \cap X^c) \mid X=k,B=j \right) =\\
\displaystyle\frac{(\# F_{\edge (t) } + j ) \, (d-k) }{d - \# C_e -  \# \{ t' \in A \mid \edge (t') = \ar e \} +j}
\end{eqnarray}
Averaging over $k$ and then $j$ gives the result.
\end{proof}
In practice, it is also useful to group current tasks together according to their directed edges. This is reflected in the following proposition:
\begin{proposition}
\label{prop:batches}
Let $A \cup A'$ denote the set of tasks, where $A'$ consists of all tasks using the directed edge $\ar{uv}.$ Then
\[
\EE_v (Q,C,A\cup A', F) = \EE_v (Q,C,A,F) + \alpha \, \# A' \, \left(\displaystyle\frac{d-\EE_v (Q,C,A,F)}{d-\#C_e} \right).
\]
\end{proposition}
\begin{proof}
Let $A'=\{ t_1, \ldots , t_k \}$ and consider the events
\[
E_i = \text{ ``} t_i \text{ finds a solution unknown after completing all tasks in } A \text{.''} 
\]
Then
\begin{align*}
\EE_v (Q,C,A\cup A', F) &= \EE_v (Q,C,A,F) + \displaystyle\sum_{i=1}^k \Pr (E_i)\\
&= \EE_v (Q,C,A,F)+ \alpha \, \# A' \, \Pr (\sol_{t_1} \notin \sol_A) \\
&= \EE_v (Q,C,A,F) + \alpha \, \# A' \,\times \\
&\displaystyle\frac{d- \EE_v (Q,C,A,F) - \Big( \frac{\# F_{\edge (t)} \cdot (d-\EE_v (Q,C,A,F))}{d-\#C_e}\Big)}{d-\#C_e-\#F_{\edge (t)}}  \\
&= \EE_v (Q,C,A,F) + \alpha \, \# A' \, \times \\
& \left(\displaystyle\frac{d-\EE_v (Q,C,A,F)}{d-\#C_e} \right).
\end{align*}
\end{proof}


%% file: algorithms.tex
For every edge $\ar{e}$ we have its \defn{potential} $\dEN_{\ar e}(x)$ at state $x=(Q,C,A,F)$. The potential guides edge selection in our main algorithm below.
Following the study in \S\ref{sec:task-selection}, the natural greedy potential aiming to maximize the expected \emph{total} number of discovered solutions is 
$$
\dEN_{\ar e}^{\EN}(x) = \EN(Q,C, A\cup \{(s,\ar e)\}, F) - \EN(x).
$$

\begin{algorithm}[Main algorithm]\label{alg:main} The following is executed on all available threads after initializing the state $x=(Q,C,A,F)$.
\noindent
\begin{algorithmic}[0]
\smallskip\hrule\smallskip
\WHILE{$\nexists Q_v$ such that $|Q_v|=d$}
\STATE Pick an edge $\ar e = (w,v)$ that maximizes $\dEN_{\ar e}(x)$ and such that there is $s\in Q_{w} \setminus \pi_wC_{e}$.
\STATE $t \gets (s,\ar e)$
\STATE Update the state: $x\gets(Q,C,A\cup\{t\},F)$.  
\STATE Update $\dEN(x)$.
\STATE Run homotopy continuation for task $t$.
\IF{the run fails}
\STATE $F_{\edge(t)} \gets F_{\edge(t)} \cup \{s\}$.  
\ELSE 
\STATE Update $Q_v$, $C_e$, and $\dEN$. \COMMENT{Note that an update is needed only for $\dEN_{\ar{e'}}$ such that $\head(\ar{e'}) = v$.} 
\ENDIF
\ENDWHILE
\smallskip\hrule\smallskip
\end{algorithmic}
\end{algorithm}

Here are other (heuristic) potential functions we considered:
: 
\begin{eqnarray*}
\dEN_{\ar e}^{\text{ord}}(x) &= &1/i, \\
&&\text{where } v_i = \head{\ar e} \quad \text{(assuming $i\in\{1,\dots,N\}$)}\\   
\dEN_{\ar e}^\omega(x) &=  &\sum_{v\in V}\omega(v)\, \left(\EN_v(Q,C, A\cup \{(s,\ar e)\},F) - \EN_v (x)\right), \\
&&\text{where weight }\omega(v)\in [0,1].
\end{eqnarray*}

Note $\dEN^{\text{ord}}$ is designed to bias edge selection towards nodes in their order of appearance. This potential is likely to force the algorithm to complete the solution set of the first node.

The weighted potential $\dEN^\omega$ depends on the design of the weight function. See \S\ref{sec:experiment-potentials} for a family of weight functions that seems to be useful in practice. 

In the initial idle state, $\dEN$ can be computed and stored and then updated during the computation. 
Propositions~\ref{prop2} and~\ref{prop_fail_pot} allow for an efficient way to do that in both sequential and parallel setting, with or without an assumption of failures. 


%% file: thresh.tex
With assumptions~\ref{randomAssumption} and~\ref{failureAssumption}, suppose we
have a complete multigraph on $N$ nodes with $m$ edges connecting each pair of
nodes, $d$ solutions at each node, and tracking success probability $\alpha $.
To each homotopy graph, we associate a \emph{solution (multi)graph} whose
vertices are given by all of the solutions at each node and whose edges are the
successful correspondences between solutions. One possible instance of this
random solution graph is depicted in Figure~\ref{fig:multigraph}. 
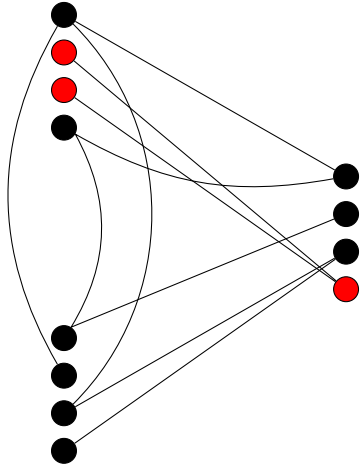
\begin{figure}
  \centering
  \begin{tikzpicture}[every node/.style={draw,shape=circle,fill=black},scale=1/4]
\node (11) at (10,0) {};
\node[fill=red] (21) at (-5,8.6) {};
\node (31) at (-5,-8.6) {};
\node (12) at (10,2) {};
\node[fill=red] (22) at (-5,8.6+2) {};
\node (32) at (-5,-8.6+2) {};
\node (13) at (10,4) {};
\node (23) at (-5,8.6+4) {};
\node (33) at (-5,-8.6+4) {};
\node[fill=red] (14) at (10,-2) {};
\node (24) at (-5,8.6-2) {};
\node (34) at (-5,-8.6-2) {};
\draw (23) to [bend right =30] (32);
\draw (13) to [bend left= 20] (24) to [bend left =30] (33) to (12);
\draw (11) -- (31) to [bend right = 45] (23);
\draw (11) -- (34);
\draw (13) -- (23);
\draw (22) -- (14) -- (21);
\end{tikzpicture}
\caption{Solution multigraph w/ $N=3, d=4,$ $m=2.$}
\label{fig:multigraph}
\end{figure}

Note that the graph in Figure~\ref{fig:multigraph} has only $10$ edges out of a possible $24.$ Nevertheless, our algorithm succeeds in completing the bottom node whenever we start from one of the $9$ black solutions, which form a large connected component. We see that connectivity of the solution graph is sufficient, but not necessary, for our algorithm to terminate.

In our random solution graph model, define $A:= A_{m,N,d}$ to be the event that the algorithm starting at a random node terminates with $d$ solutions. We are interested in the asymptotic behavior of $\Pr (A)$ as $d\to \infty ,$ with reasonable assumptions on $m$ and $N.$ More precisely, we wish to describe an interval $[a_{m,N,d},b_{m,N,d}]$ containing a \emph{threshold} for the event $A;$ this means that for $\alpha (m,N,d) = o(a_{m,N,d}),$ we have $\Pr (A_{m,N,d})\to 0,$ while $\Pr (A_{m,N,d} )\to 1$ if $\alpha (m,N,D) = \omega (b_{m,N,d}).$

The characterization of thresholds for various properties is a well-studied problem in random graph theory, particularly in the context of the Erd\"{o}s-Renyi graph model. Our random solution graph does not enjoy the same asymptotic properties as the Ed\"{o}s-Renyi graph---since no two solutions at the same node may be connected, the graph is sparse, even for $\alpha $ near $1.$ Minding these difficulties, we provide a simple threshold region for the event $A$ in Proposition~\ref{prop:threshold_bounds}---see subsection~\ref{subsec:parth_failures} for experimental verification and further discussion.
\begin{proposition}
\label{prop:threshold_bounds}
With $m, N$ possibly depending on $d,$ we have the following large-$d$ asymptotics:
\begin{itemize}
\item[i)] If $\alpha (d) = o \left( (N m)^{-1}\right)$ and $N(d) = o\Big(\exp(d)\Big),$ then \\
$\lim_{d\to \infty } \Pr (A_{m,N,d}) =0.$
\item[ii)] If $\alpha (d) = \omega \left(\log d / m \right)$ and $N(d)= O(\log d),$ then \\
$\lim_{d\to \infty } \Pr (A_{m,N,d}) =1.$
\end{itemize}
\end{proposition}
We require a simple fact known as the Harris/Kleitman inequality, specialized to our model (cf.~\cite{alon2004probabilistic} pp.~86-87, ~\cite{bollobas2006percolation} pp.~39-41):
\begin{theorem}[Harris/Kleitman Inequality]
\label{harris}
If $\mathcal{A}$ and $\mathcal{B}$ are events in the random solution graph model which are upward-closed with respect to inclusion,
\[
  \label{eq:crude_bounds}
\Pr (\mathcal{A} \cap \mathcal{B}) \ge \Pr (\mathcal{A} ) \times \Pr (\mathcal{B}).
\]
\end{theorem}
In random graph theory, a property which is upward-closed with respect to inclusion is called a \emph{monotone increasing property.} For us, monotone increasing simply means that increasing $\alpha $ does not decrease $\Pr (\mathcal{A})$ or $\Pr (\mathcal{B}).$ It is a famous result that every monotone property in the Erd\"{o}s-Renyi model has a sharp threshold---for a precise statement, see \cite{Bol1987}.
\begin{proof}[Proof of Proposition~\ref{prop:threshold_bounds}]
Consider the following auxiliary events:
\begin{itemize}
\item $S:= S_{m,N,d,\alpha }$ will denote the event that there exists some node with a successful correspondence at each solution
\item $C := C_{m,N,d,\alpha}$ will denote the event that the solution graph is connected
\end{itemize}
Clearly we have
\begin{equation}
  \label{crude_bounds}
\Pr (C) \le \Pr (A) \le \Pr (S)
\end{equation}
For part i), we may assume WLOG that $\alpha (d) >0$ for $d$ sufficiently large. Now, simply note that
\begin{align*}
\Pr (S) &= 1 - \Pr (\text{ all nodes fail } ) \\
&\le 1 - \displaystyle\prod_{j=1}^N \Pr (\text{node } j \text{ fails}) \tag{Theorem~\ref{harris}} \\
&\le 1 - \Big(1- (\alpha N m)^d\Big)^N .
\end{align*}
For $N=O(1)$ as $d\to \infty ,$ we have
\[
(1 - (\alpha N m)^d)^N= (1 - o(1)^d)^N \to 1 \text{ as } d\to \infty . 
\]
For the regime $\omega (1) = N(d) = o(\exp (d)) ,$ we have
\[
(1 - (\alpha N m)^d)^N \sim \exp (-(\alpha  N^{1+1/d} m)^d ), 
\]
which is $\omega (1)$ for $\alpha (d) = o(N^{-(1+1/d)} \, m^{-1} ) = o ((Nm)^{-1}).$
In either regime, we have $\Pr (S)\to 0$ as $d\to \infty .$\\\\
To bound $\Pr (A)$ from below, let $v_1, \ldots , v_n$ be the nodes of the homotopy graph and $G_i$ denote the subgraph of the solution graph induced by the solutions at nodes $v_1$ through $v_i.$ By repeated application of Theorem~\ref{harris}, we have
\begin{align*}
\Pr (C_N) &= \Pr ( C_{N-1} \cap \textrm{ each sol at } v_N \textrm{ has a nbr in } G_{N-1})\\
&\ge \Pr ( C_{N-1}) \times \Pr (\textrm{each sol at } v_N \textrm{ has a nbr in } G_{N-1})\\
&\ge \displaystyle\prod_{i=1}^N \Pr(\textrm{ all solutions at } v_i \textrm{ have a nbr in } G_{i-1} )\\
&= \displaystyle\prod_{i=1}^N \Big( 1 - (1-\alpha )^{(N-i) e } \Big)^d \\
&\ge \Big( 1 - \exp (-\alpha m )\Big)^{Nd} \\
\end{align*}
Now, setting 
\[
\alpha (d) = \displaystyle\frac{\log (d) \times \bigg( 1 + \log_d \big(N + g(d)\big)\bigg)}{m},
\]
with $g(d)$ \emph{any} function such that $g(d)\to \infty $ as $d\to \infty ,$ we have
\[
\Pr (C) \ge (1 - g(d)/(Nd))^{Nd} \sim \exp \big( -g(d) \big) \to 1.
\]
\end{proof}

%% file: experiments.tex
Our simulator enables the study of two types of experiments:
\begin{itemize}
\item
Experiments analyzing fake solution graph data generated according to Assumptions~\ref{randomAssumption} and~\ref{failureAssumption}, and 
\item 
experiments based on real parametric systems, for which all data --- actual solutions, actual correspondences for edges in the graph, actual timings for each homotopy path that may be tracked --- is harvested \emph{before} the experimentation begins.
\end{itemize}

The simulator (code available at \url{https://github.com/sommars/parallel-monodromy}) proceeds in two stages:
\begin{itemize}
\item
The first stage takes either randomly generated data using Assumption~\ref{randomAssumption},
which does not require running homotopy continuation, or collects the data through
tracking homotopy paths with existing software. 
\item 
The second uses the datafile produced by the first. 
If several threads are simulated then we assume that there is no communication overhead, which is a close approximation of reality. Indeed, the messages passed around are rather short: a longest one contains coordinates of a newly discovered solution. This cost can be ignored in comparison to the cost of a homotopy continuation task.
\end{itemize}

From observed runs of \verb|PHCpack|~\cite{V99} and 
\verb|NumericalAlgebraicGeometry|~\cite{Leykin:NAG4M2}, we chose to model the time taken by each fake path track on the negative binomial distribution with parameters  $p = 0.3,$ $n = 10$.
For clarity and consistency with the results of \S\ref{sec:threshold}, all simulations have been run using the complete graph configuration
described in~\cite{firstmonodromypaper}.

\subsection{Parallel Performance}

To demonstrate the quality of a parallel algorithm, the typical
metrics used are \textit{speedup} and \textit{efficiency} (for textbook references,
see~\cite{parallel-book2},~\cite{parallel-book1}). For a number
of processors $p$, speedup is defined
to be
\begin{equation}
S(p) = \frac{\text{sequential execution time}}{\text{parallel execution time}}
\end{equation}
while efficiency is defined as
\begin{equation}
E(p) = \frac{S(p)}{p}\times 100\%
\end{equation}
Ideally one would obtain $S(p) = p$ and $E(p) = 100\%$, which means that
all processor resources are constantly in use and no extra work is performed,
compared to running the program with a single processor.

We ran two experiments to observe the efficiency of our algorithm, one with
simulated data as in (1) and one with observed
data as in (2). Table~\ref{FakeParallelEfficiency} contains efficiency results
for the simulated data experiment, while Table~\ref{CyclicParallelEfficiency}
has efficiency results for the cyclic-$n$ roots problem.

\input{parallelSimulatedTrial.tex}

\input{parallelCyclicTrial.tex}

\noindent The cyclic $n$-roots problem is a classic
benchmark problem in polynomial system solving, commonly formulated as
\begin{equation} \label{eqcyclicsys}
   \begin{cases}
   x_{0}+x_{1}+ \cdots +x_{n-1}=0 \\
   i = 2, 3, \ldots, n-1: 
    \displaystyle\sum_{j=0}^{n-1} ~ \prod_{k=j}^{j+i-1}
    x_{k~{\textrm mod}~n}=0 \\
   x_{0}x_{1}x_{2} \cdots x_{n-1} - 1 = 0. \\
\end{cases}
\end{equation}

\noindent Both Tables~\ref{FakeParallelEfficiency} and~\ref{CyclicParallelEfficiency}
show the same relationships: as the number of threads increases, efficiency
slowly decreases, and as the size of the problem increases,
efficiency improves. This shows that it is an effective algorithm
for running large systems in parallel, though it is unfortunate that for huge
numbers of threads that efficiency decreases.

One could be concerned that Algorithm~\ref{alg:main} would be slow to start,
because initially a single node has a single solution. For small homotopy graphs
with large numbers of threads, some threads will by necessity be idle until
there are sufficiently many tasks available. Define
\begin{equation}
\%\text{Idle} = \frac{\sum_{i=1}^p\text{Idle time}}{\text{Wall time}\times p}.
\end{equation}

\noindent As the number of solutions increases, $\%\text{Idle}$ approaches
zero. It would be possible to make $\%\text{Idle} = 0$
through a modification to Algorithm~\ref{alg:main}.
When a thread rests idle waiting for a task to become available,
it could define its own edge by picking a random $\gamma$
and tracking the sole known solution to a different node. In doing this,
it has the potential to discover new solutions, but without adding
to the known set of correspondences. Each thread could do this until it
can be assigned a path track as the algorithm prescribes. However, this
will provide only a minimal benefit, because the amount of idle time according to
a wall clock is low.

\subsection{Path Failures}
\label{subsec:parth_failures}

The ``probability-one'' homotopy in a linear family $F_p$ fails with some nonzero probability. At fixed precision, this probability becomes non-negligible, say, as the degree of the discriminant rises. In practice, the reliability of homotopy continuation may be impacted by more \emph{aggressive path-tracking.} For instance, raising the minimum step-size lowers the number of predictor steps, but there is a risk that errors accumulated may too large to finish continuation. In MS, this risk is spread across its incoming edges. Thus, we are interested in balancing tradeoffs between task reliability and speed.

Assumption~\ref{failureAssumption} gives a simple model for path failures in a practical setting. An important feature of this model is that our simulator assumes a ``true correspondence'' between the solutions of two connected nodes before declaring that some of these paths fail. Thus, our model of failures ignores the phenomenon of path-jumping (potentially resulting in a 2-1 correspondence between approximate solutions,) or the possibility that some node has a near-singular solution.  A logical next step would be to incorporate these possibilities into our model. However, we find that the simple model already sheds some light on the tradeoff previously described.

\begin{figure}
  \centering
  \includegraphics[width=8cm,height=7cm]{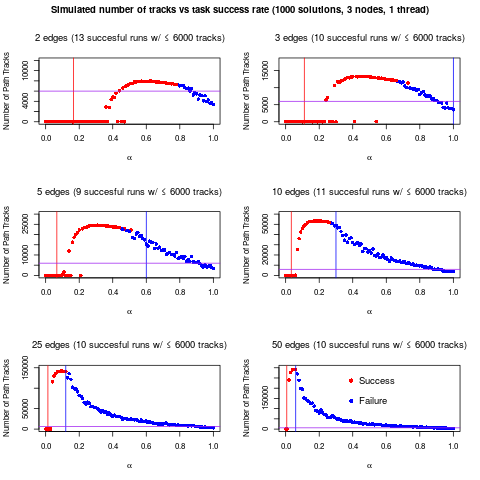}
  \caption{Total number of tracks vs task success rate $\alpha $ for varying edge multiplicities on a 3-node graph with 1000 solutions. The red and blue vertical lines given by $x=1/3m$ and $x=\log_{10} (d) / m,$ respectively, give an approximate window for the failure threshold. To give a sense of scale, a purple horizontal line at $y=6000$ has been added.}
    \label{fig:sim_runtime_vs_edges}
\end{figure}

The plots in figure~\ref{fig:sim_runtime_vs_edges} supplement the results of Section~\ref{sec:threshold}. In each panel, the vertical distance equals the theoretical maximum number of tracks for each graph layout. Each run was performed with a single fake thread using the potential {\tt potE}. These plots illustrate a major strength of using potentials in the presence of failures---even when additional edges are added, the number of path tracks at a \emph{fixed} failure rate is stable (eg.~at most 6000 for $\alpha \ge 0.9.$)

The bounds in Proposition~\ref{prop:threshold_bounds} do not provide a useful upper bound on the threshold of global failure when the number of edges is relatively small (as in the top two plots of Figure~\ref{fig:sim_runtime_vs_edges}.) We attempted to determine tighter threshold regions experimentally---see Table~\ref{table:alpha-threshold-simulated} for fabricated data and Table~\ref{table:alpha-threshold-cyclic} for the cyclic $n$-roots problem.

\input{random-alpha-threshold.tex}
\input{cyclic-alpha-threshold.tex}

\subsection{Potential functions and edge selection}\label{sec:experiment-potentials}

We defined potentials $\dEN^{\EN}$, $\dEN^{\text{ord}}$, and $\dEN^\omega$. The last potential offers a lot of freedom to the user of the method. 
For instance, we could combine the ideas behind $\dEN^{\EN}$ and $\dEN^{\text{ord}}$ in $\dEN^\omega$ by setting 
\begin{equation}\label{eq:lambda-weight}
\omega(v) = (|Q_v| / d)^\lambda,\quad \lambda\geq 0,
\end{equation}
where $d$ is the root count. (It could be replaced with the maximal number of solutions known at any node).
Note that if $\lambda=0$, one gets $\dEN^{\EN}$; for large $\lambda$ the effect is similar to that of $\dEN^{\text{ord}}$ except the nodes are likely to be ordered according to the number of known solutions at any point of the execution. 

In the sequential case, \cite{firstmonodromypaper} shows that edge selection guided by the greedy potential $\dEN^{\EN}$ outperforms several naive choices, among them the random edge selection strategy. According to our experiments this, as we expect, still holds for the parallel setting.

We conducted several experiments with the weight potential $\dEN^\omega$ on graphs with edge multiplicity $m=1$ for fabricated and cyclic problems of degree up to $10000$ with and without failures. The weights described in (\ref{eq:lambda-weight}) seem to deliver better (but not necessarily the best) performance as $\lambda\to\infty$. In other words, while a variant of the order potential $\dEN^{\text{ord}}$ may serve as a good heuristic, there is still some room for improvement for edge selection strategies guiding the MS algorithm.

%% file: parallelSimulatedTrial.tex
\begin{table}[h]
\centering
\begin{tabular}{l||c|c|c|c|c|}
\#Solutions & 100 & 500 & 1000 & 5000 & 10000\\ \hline
1 & 100\% & 100\% & 100\% & 100\% & 100\%\\
2 & 98.87\% & 98.36\% & 99.88\% & 98.61\% & 99.3\%\\
4 & 96.71\% & 96.34\% & 98.28\% & 99.75\% & 100.45\%\\
8 & 91.92\% & 95.04\% & 97.55\% & 98.7\% & 100.56\%\\
16 & 84.65\% & 92.82\% & 98.68\% & 99.24\% & 99.82\%\\
32 & 71.39\% & 87.12\% & 94.89\% & 97.8\% & 100.74\%\\
64 & 55.04\% & 78.78\% & 89.45\% & 96.7\% & 99.07\%\\
128 & 35.95\% & 65.82\% & 79.62\% & 93.68\% & 97.87\%\\
\end{tabular}
\caption{Efficiency for simulated polynomial systems with varied numbers of solutions.}
\label{FakeParallelEfficiency}
\end{table}

%% file: parallelCyclicTrial.tex
\begin{table}[h]
\centering
\begin{tabular}{l||c|c|c|c|c|c|}
$n$ & 5 & 6 & 7 & 8 & 9 & 10\\ \hline
1 & 100\% & 100\% & 100\% & 100\% & 100\% & 100\%\\
2 & 110.48\% & 98.34\% & 104.3\% & 99.41\% & 99.44\% & 109.02\%\\
4 & 107.7\% & 98.57\% & 110.79\% & 103.06\% & 99.81\% & 107.62\%\\
8 & 101.53\% & 98.23\% & 108.02\% & 108.59\% & 101.02\% & 106.58\%\\
16 & 94.88\% & 91.52\% & 103.53\% & 100.53\% & 101.79\% & 103.91\%\\
32 & 76.23\% & 86.73\% & 97.72\% & 100.81\% & 101.92\% & 105.54\%\\
64 & 54.59\% & 70.47\% & 93.45\% & 98.62\% & 99.92\% & 102.88\%\\
128 & 34.38\% & 52.37\% & 84\% & 96.23\% & 97.81\% & 102\%\\
\end{tabular}
\caption{Efficiency for cyclic-$n$ polynomial systems.}
\label{CyclicParallelEfficiency}
\end{table}

%% file: random-alpha-threshold.tex
\begin{table}[h]
\centering
\begin{tabular}{l||c|c|c|c|c|c|}$d\setminus N$&
4&
5&
6&
7&
8&9\\ \hline
16&
.716&
.544&
.426&
.36&
.332&
.271\\
32&
.756&
.599&
.495&
.427&
.362&
.312\\
64&
.771&
.62&
.537&
.47&
.391&
.366\\
128&
.799&
.666&
.584&
.498&
.453&
.405\\
256&
.841&
.732&
.634&
.572&
.497&
.445\\
512&
.873&
.752&
.674&
.598&
.536&
.49
\end{tabular}
\caption{
An approximate threshold for the success rate \ensuremath{\alpha}.
\ensuremath{N} = the number of nodes in the complete graph (with \ensuremath{m=1}), 
\ensuremath{d} = the number of solutions.    
}\label{table:alpha-threshold-simulated}
\end{table}

%% file: cyclic-alpha-threshold.tex
\begin{table}[h]
\centering
\begin{tabular}{l||c|c|c|c|c|c|}$n\setminus N$&
5&
6&
7&
8&9\\ \hline
5&
.546&
.492&
.34&
.298&
.281\\
6&
.605&
.516&
.416&
.344&
.316\\
7&
.686&
.611&
.531&
.452&
.453\\
8&
.734&
.688&
.647&
.564&
.492\\
9&
.818&
.733&
.672&
.629&
.556
\end{tabular}
\protect\caption{
An approximate threshold for the success rate \ensuremath{\alpha} for the cyclic-\ensuremath{n} family.
\ensuremath{N} = the number of nodes in the complete graph (with \ensuremath{m=1}).
}
\protect\label{table:alpha-threshold-cyclic}
\end{table}

%% file: conclusion.tex
The benefits of the Monodromy Solver framework are demonstrated by an implementation in Macaulay2~\cite{www:MonodromySolver,M2www}, which outperforms all existing blackbox polynomial system solvers on certain classes of problems. This is reported in \S 6.4 of the first article devoted to the framework~\cite{firstmonodromypaper}. 

The present work addressed items 1 (failures) and 3 (parallelization) in the program outlined in \S 7 of~\cite{firstmonodromypaper}. The experiments conducted with the simulator that we built, albeit not very extensive, shed light on the phenomena arising with the introduction of failures and parallel computation. The results of the experiments and the simulator itself will help to hone the core of the technique as well as construct efficient heuristics for software implementation in the future.

%% file: art.bbl
\begin{thebibliography}{10}

\bibitem{alon2004probabilistic}
N.~Alon and J.~H Spencer.
\newblock {\em The probabilistic method}.
\newblock John Wiley \& Sons, 2004.

\bibitem{bollobas2006percolation}
B.~Bollob{\'a}s and O.~Riordan.
\newblock {\em Percolation}.
\newblock Cambridge University Press, 2006.

\bibitem{Bol1987}
B.~Bollob{\'a}s and A.~G. Thomason.
\newblock Threshold functions.
\newblock {\em Combinatorica}, 7(1):35--38, Mar 1987.

\bibitem{www:MonodromySolver}
T.~Duff, C.~Hill, A.~Jensen, K.~Lee, A.~Leykin, and J.~Sommars.
\newblock {MonodromySolver:} a {Macaulay2} package for solving polynomial
  systems via homotopy continuation and monodromy.
\newblock Available at
  http://people.math.gatech.edu/$\sim$aleykin3/MonodromySolver.

\bibitem{firstmonodromypaper}
T.~Duff, C.~Hill, A.~Jensen, K.~Lee, A.~Leykin, and J.~Sommars.
\newblock Solving polynomial systems via homotopy continuation and monodromy.
\newblock {\em To appear in IMA Journal of Numerical Analysis}, 2018.

\bibitem{M2www}
D.~R. Grayson and M.~E. Stillman.
\newblock Macaulay2, a software system for research in algebraic geometry.
\newblock Available at http://www.math.uiuc.edu/Macaulay2/.

\bibitem{Gunji2006}
T.~Gunji, S.~Kim, K.~Fujisawa, and M.~Kojima.
\newblock Phompara -- parallel implementation of the polyhedral homotopy
  continuation method for polynomial systems.
\newblock {\em Computing}, 77(4):387--411, 2006.

\bibitem{Harimoto:1987:GHA:645818.669226}
S.~Harimoto and L.T. Watson.
\newblock The granularity of homotopy algorithms for polynomial systems of
  equations.
\newblock In {\em Proceedings of the Third SIAM Conference on Parallel
  Processing for Scientific Computing}, pages 115--120, Philadelphia, PA, USA,
  1989. Society for Industrial and Applied Mathematics.

\bibitem{hauenstein2011regeneration}
J.~Hauenstein, A.~Sommese, and C.~Wampler.
\newblock Regeneration homotopies for solving systems of polynomials.
\newblock {\em Mathematics of Computation}, 80(273):345--377, 2011.

\bibitem{hauenstein2016numerical}
J.D. Hauenstein, J.I. Rodriguez, and F.~Sottile.
\newblock Numerical computation of galois groups.
\newblock {\em arXiv preprint arXiv:1605.07806}, 2016.

\bibitem{HuberSturmfels:PolyhedralHomotopies}
B.~Huber and B.~Sturmfels.
\newblock A polyhedral method for solving sparse polynomial systems.
\newblock {\em Math. Comp.}, 64(212):1541--1555, 1995.

\bibitem{parallel-book2}
D.B. Kirk and W.W. Hwu.
\newblock {\em Programming Massively Parallel Processors: A Hands-on Approach}.
\newblock Morgan Kaufmann Publishers Inc., San Francisco, CA, USA, 1st edition,
  2010.

\bibitem{Leykin:NAG4M2}
A.~Leykin.
\newblock Numerical algebraic geometry.
\newblock {\em The Journal of Software for Algebra and Geometry}, 3:5--10,
  2011.

\bibitem{Leykin:2007:CMV:1278177.1278195}
A.~Leykin and F.~Sottile.
\newblock Computing monodromy via parallel homotopy continuation.
\newblock In {\em Proceedings of the 2007 International Workshop on Parallel
  Symbolic Computation}, PASCO '07, pages 97--98, New York, NY, USA, 2007. ACM.

\bibitem{Leykin-Sottile:HoG}
A.~Leykin and F.~Sottile.
\newblock Galois groups of {S}chubert problems via homotopy computation.
\newblock {\em Math. Comp.}, 78(267):1749--1765, 2009.

\bibitem{1488691}
A.~Leykin and J.~Verschelde.
\newblock Factoring solution sets of polynomial systems in parallel.
\newblock In {\em 2005 International Conference on Parallel Processing
  Workshops (ICPPW'05)}, pages 173--180, June 2005.

\bibitem{Ley-Ver-new-monodromy-05}
A.~Leykin and J.~Verschelde.
\newblock Decomposing solution sets of polynomial systems: a new parallel
  monodromy breakup algorithm.
\newblock {\em International Journal of Computational Science and Engineering},
  4(2):94--101, 2009.

\bibitem{Leykin2006}
A.~Leykin, J.~Verschelde, and Y.~Zhuang.
\newblock {\em Parallel Homotopy Algorithms to Solve Polynomial Systems}, pages
  225--234.
\newblock Springer Berlin Heidelberg, Berlin, Heidelberg, 2006.

\bibitem{Morgan87}
A.~Morgan.
\newblock {\em Solving polynomial systems using continuation for engineering
  and scientific problems}.
\newblock Prentice Hall Inc., Englewood Cliffs, NJ, 1987.

\bibitem{MORGAN19891339}
A.P. Morgan and L.T. Watson.
\newblock A globally convergent parallel algorithm for zeros of polynomial
  systems.
\newblock {\em Nonlinear Analysis: Theory, Methods \& Applications},
  13(11):1339 -- 1350, 1989.

\bibitem{SVW2001:monodromy}
A.~J. Sommese, J.~Verschelde, and C.~W. Wampler.
\newblock {\em Using Monodromy to Decompose Solution Sets of Polynomial Systems
  into Irreducible Components}, pages 297--315.
\newblock Springer Netherlands, Dordrecht, 2001.

\bibitem{Sommese-Wampler-book-05}
A.J. Sommese and C.W. Wampler.
\newblock {\em The numerical solution of systems of polynomials}.
\newblock World Scientific Publishing Co. Pte. Ltd., Hackensack, NJ, 2005.

\bibitem{V99}
J.~Verschelde.
\newblock Algorithm 795: {PHC}pack: A general-purpose solver for polynomial
  systems by homotopy continuation.
\newblock {\em ACM Trans. Math. Softw.}, 25(2):251--276, 1999.
\newblock Available at {http://www.math.uic.edu/$\sim$jan}.

\bibitem{Verschelde-Verlinden-Cools}
J.~Verschelde, P.~Verlinden, and R.~Cools.
\newblock Homotopies exploiting newton polytopes for solving sparse polynomial
  systems.
\newblock {\em SIAM J. Numer. Anal.}, 31(3):915--930, June 1994.

\bibitem{1690736}
J.~Verschelde and Y.~Zhuang.
\newblock Parallel implementation of the polyhedral homotopy method.
\newblock In {\em 2006 International Conference on Parallel Processing
  Workshops (ICPPW'06)}, pages 8 pp.--488, 2006.

\bibitem{WAMPLER1992143}
C.W. Wampler.
\newblock Bezout number calculations for multi-homogeneous polynomial systems.
\newblock {\em Applied Mathematics and Computation}, 51(2):143 -- 157, 1992.

\bibitem{parallel-book1}
B.~Wilkinson and M.~Allen.
\newblock {\em Parallel Programming: Techniques and Applications Using
  Networked Workstations and Parallel Computers (2nd Edition)}.
\newblock Prentice-Hall, Inc., Upper Saddle River, NJ, USA, 2004.

\end{thebibliography}
